\documentclass[11pt,letter]{article}
\usepackage{bbding}
\fussy
\usepackage{framed, xcolor}
\usepackage{tocvsec2}
\usepackage{datetime}
\usepackage{pifont}
\usepackage{tikz}
\usepackage{pdflscape}
\usepackage{subfigure}          
\usepackage{colortbl}
\usepackage{booktabs}
\usepackage{pdfsync}
\usepackage[font=scriptsize,bf]{caption}
\usepackage{tikz,subfigure}
\usepackage[active]{srcltx}
\usepackage[margin=1in]{geometry}
\usepackage{algorithm,algorithmicx}
\usepackage[noend]{algpseudocode}
\usepackage{datetime}
\usepackage{epsfig,amssymb,amsfonts,amsmath,amsthm}
\usepackage{multirow}
\usepackage[numbers,sort&compress,sectionbib]{natbib}
\usepackage{amsfonts}
\usepackage{xspace}
\usepackage{tabularx}
\usepackage{pstricks}
\usepackage{setspace}
\usepackage{xcolor}

\usepackage{rotating}
\usepackage{minitoc}
\usepackage{enumerate}
\definecolor{ocre}{rgb}{0.72,0,0} 
\definecolor{newblue}{rgb}{0.2,0.2,0.6} 

\usepackage[citecolor=newblue,colorlinks=true,linkcolor=ocre]{hyperref}

\usepackage[normalem]{ulem}
\usepackage{fancybox}

\newenvironment{fminipage}%
  {\begin{Sbox}\begin{minipage}}%
  {\end{minipage}\end{Sbox}\fbox{\TheSbox}}

\newcommand{\remove}[1]{}

\newcommand{\rot}{\intercal}
\newcommand{\ce}{\mathrm{e}}
\newcommand{\tr}{\mathrm{tr}}

\renewcommand{\deg}{\mathrm{deg}}

\newcommand{\eps}{\epsilon}

\renewcommand{\leq}{\leqslant}
\renewcommand{\geq}{\geqslant}
\renewcommand{\le}{\leqslant}

\newcommand{\thmref}[1]{Theorem~\ref{thm:#1}}

\newcommand{\lemref}[1]{Lemma~\ref{lem:#1}}

\newcommand{\secref}[1]{Section~\ref{sec:#1}}

\newcommand{\eq}[1]{\eqref{eq:#1}}

\newcommand{\Pro}[1]{\mathbb{P} \left[\,#1\,\right]}

\newcommand{\TEx}[1]{\widetilde{\mathbb{E}} \left[\,#1\,\right]}
\newcommand{\Ex}[1]{\mathbb{E} \left[\,#1\,\right]}

\renewcommand{\tilde}{\widetilde}
\renewcommand{\epsilon}{\varepsilon}

\definecolor{ocre}{RGB}{150,22,11} 

\usepackage{setspace}

\newtheorem{thm}{Theorem}[section]  

\newtheorem{lem}[thm]{Lemma}

\newtheorem{cor}[thm]{Corollary}

\newtheorem*{rem*}{Remark}
\newtheorem{pro}[thm]{Proposition}

\newtheorem{assumption}[thm]{Assumption}

\renewcommand{\tilde}{\widetilde}
\usepackage{todonotes}

\numberwithin{equation}{section}

\renewcommand{\vec}[1]{#1}
\newcommand{\mat}[1]{#1}

\title{{\bf Constructing Linear-Sized Spectral Sparsification\\
in Almost-Linear Time}}

\author{Yin Tat Lee\\
MIT\\ Cambridge, USA\\
\texttt{yintat@mit.edu}
\and
He Sun\\
The University of Bristol\\
Bristol, UK\\
\texttt{h.sun@bristol.ac.uk}
}

\author{Yin Tat Lee\\
MIT\\ Cambridge, USA\\
\texttt{yintat@mit.edu}
\and
He Sun\\
The University of Bristol\\
Bristol, UK\\
\texttt{h.sun@bristol.ac.uk}
}

\date{}

\begin{document}

\maketitle

\begin{abstract}
We present the first almost-linear time algorithm for constructing linear-sized spectral sparsification for graphs.  This improves all previous constructions of linear-sized spectral sparsification, which requires $\Omega(n^2)$ time~\cite{BSS12,Z12, zhu15}.

A key ingredient in our algorithm is a novel combination of two techniques used in literature for constructing spectral sparsification: Random sampling by
 effective resistance~\cite{journals/siamcomp/SpielmanS11}, and adaptive construction based on barrier functions~\cite{BSS12, zhu15}.
 
\vspace{0.5cm}

 \textbf{keywords:} algorithmic spectral graph theory, spectral sparsification

 \end{abstract}

\thispagestyle{empty}

\setcounter{page}{0}

\newpage

\section{Introduction}

Graph sparsification is the procedure of approximating a graph $G$ by a   sparse  graph $G'$  such that certain quantities between $G$ and $G'$ are preserved. For instance, \emph{spanners} are defined between two graphs in which the distances between any pair of vertices in these two graphs are approximately the same~\cite{Che89}; \emph{cut sparsifiers} are reweighted sparse graphs of the original graphs such that the weights of every cut between the sparsifiers and the original graphs are approximatedly the same~\cite{BK96}. Since both storing and processing  large-scale graphs are expensive, graph sparsification is  one of the most fundamental building blocks in designing fast graph algorithms, including solving Laplacian systems~\cite{spielman2004nearly,koutis2010approaching,koutis2011nearly,KoutisLP12,peng2014efficient,lee2015sparsified}, designing approximation algorithms for the maximum flow problem~\cite{BK96,KelnerLOS14,sherman2013nearly}, and solving streaming problems~\cite{kelner2013spectral,kapralov2014single}. Beyond graph problems, techniques developed for spectral sparsification are widely used in  randomized linear algebra~\cite{mahoney2011randomized,li2013iterative,cohen2015uniform}, sparsifying linear programs~\cite{lee2014path}, and  various pure mathematics problems~\cite{spielman2012elementary,srivastava2012contact,marcus2013interlacing,barvinok2014thrifty}.

In this work, we study spectral sparsification introduced by Spielman and Teng~\cite{spielman_teng:SS11}: A \emph{spectral sparsifier} is a reweighted \emph{sparse} subgraph of the original graph such that, for all real vectors, the Laplacian quadratic forms between that subgraph and the original graph are approximately the same. Formally,  for any undirected and weighted graph $G=(V,E,w)$ with $n$ vertices and $m$ edges, we call a subgraph $G'$ of $G$, with proper reweighting of the edges,  is a $(1+\varepsilon)-$\emph{spectral sparsifier} if it holds for any $x\in \mathbb{R}^n$ that
\[
(1-\eps)x^{\rot}L_Gx\leq x^{\rot}L_{G'}x\leq (1+\varepsilon) x^{\rot}L_Gx,
\]
where $L_G$ and $L_{G'}$ are the respective graph Laplacian matrices of $G$ and $G'$.

Spielman and Teng~\cite{spielman_teng:SS11} presented the first algorithm for constructing spectral sparsification. For any undirected graph $G$ of $n$ vertices, their algorithm runs in $O(n\log ^cn/\eps^2)$ time, for some big constant $c$, and produces a spectral sparsifier  with $O(n\log^{c'}n/\eps^2)$ edges for some $c'\geq 2$.  Since then, there has been a wealth of work  on spectral sparsification. For instance,
Spielman and Srivastava~\cite{journals/siamcomp/SpielmanS11} presented a nearly-linear time algorithm for constructing a spectral sparsifier of
$O(n\log n/\eps^2)$ edges.   Batson, Spielman and Srivastava~\cite{BSS12}
presented an algorithm for constructing spectral sparsifiers with $O(n/\eps^2)$ edges, which is optimal up to a  constant.
However, all previous constructions either require $\Omega\left(n^{2+\eps}\right)$ time in order to produce linear-sized sparsifiers~\cite{BSS12, Z12, zhu15}, or  $O(n\log^{O(1)} n/\eps^2)$ time but the number of edges in the sparsifiers is sub-optimal.

In this paper we present the first almost-linear time algorithm for constructing linear-sized spectral sparsification for graphs.   Our result is summarized as follows:

\begin{thm} \label{thm:graph}
Given any integer $q\geq 10$ and $0 < \eps\leq 1/120$.
Let $G=(V,E,w)$ be an undirected and weighted graph with $n$ vertices and $m$ edges.
Then, there is an algorithm that outputs a $(1+\eps)$-spectral sparsifier of $G$ with $O\left(\frac{qn}{\varepsilon^2}\right)$ edges. The algorithm runs in   $\tilde{O}\left(\frac{q\cdot m\cdot n^{5/q}}{\eps^{4+4/q}}\right)
$ time.
\end{thm}

Graph sparsification is known as  a special case of sparsifying sums of rank-1 positive semi-definite~(\textsf{PSD}) matrices~\cite{BSS12,journals/siamcomp/SpielmanS11},  and our algorithm works in this general setting as well.  Our result is summarized  as follows:

\begin{thm}\label{thm:general}
Given any integer $q\geq 10$ and $0 < \eps\leq 1/120$.
Let $I=\sum_{i=1}^m v_iv_i^{\rot}$ be the sum of $m$ rank-1 \textsf{PSD}  matrices. Then, there is an algorithm
that outputs  scalers $\{s_i\}_{i=1}^m$ with $|\{s_i: s_i\neq0\}|=O\left(\frac{qn}{\eps^2}\right)$ such that
\[
(1-\eps)\cdot I\preceq \sum_{i=1}^m s_i v_i v_i^{\rot} \preceq (1+\eps) \cdot I.
\]
The algorithm runs in $\tilde{O}\left(\frac{q  m}{\eps^2}\cdot n^{\omega-1+3/q}\right)$ time, where $\omega$ is the matrix-multiplication constant.
\end{thm}

A key ingredient in our algorithm is a novel combination of two techniques used in literature for constructing spectral sparsification: Random sampling by
 effective resistance of edges~\cite{journals/siamcomp/SpielmanS11}, and adaptive construction based on barrier functions~\cite{BSS12, zhu15}. We will present an overview of the algorithm, and the intuitions behind it in \secref{algo}.

\paragraph{Preliminaries}

Let $G=(V,E,w)$ be a connected,   undirected and weighted  graph with $n$ vertices and $m$ edges, and weight function $w: V\times V\rightarrow \mathbb{R}_{\geq 0}$. The Laplacian matrix of $G$ is an $n$ by $n$ matrix $L$ defined by
\begin{equation} \nonumber
L_G(u,v)=\left\{ \begin{aligned}
         -w(u,v) & \qquad \mbox{if $u\sim v$,} \\
         \deg(u) & \qquad \mbox{if $u=v$,} \\
                 0&\qquad \mbox{otherwise,}
                          \end{aligned} \right.
                          \end{equation}
where $\deg(u)=\sum_{v\sim u} w(u,v)$.
It is easy to see that
\[
x^{\rot}L_G x=\sum_{u\sim v} w_{u,v}(x_u-x_v)^2 \geq 0,
\]
for any $x\in \mathbb{R}^n$.

For any matrix $A$, let $\lambda_{\max}(A)$ and $\lambda_{\min}(A)$ be the maximum and minimum eigenvalues of $A$. The condition number of matrix $A$ is defined by $\lambda_{\max}(A)/\lambda_{\min}(A)$.
For any two matrices $A$ and $B$, we write $A\preceq B$ to represent $B-A$ is positive semi-definite~(\textsf{PSD}), and $A\prec B$ to represent $B-A$
 is positive definite. For any two matrices $A$ and $B$ of equal dimensions, let $A\bullet B\triangleq \tr\left(A^{\rot}B\right)$. 
For any function $f$, we write $\tilde{O}(f) \triangleq O(f\cdot\log^{O(1)}f)$. For matrices $A$ and $B$, we write $A\approx_{\eps} B$ if
$
(1-\eps)\cdot A\preceq B\preceq (1+\eps) A$.

\section{Algorithm\label{sec:algo}}
We  study the algorithm of sparsifying the sum of rank-1  \textsf{PSD} matrices in this section.
Our  goal is to, for any vectors $v_1,\cdots v_m$ with $\sum_{i=1}^m v_iv_i^{\rot}=I$, find scalars $\{s_i\}_{i=1}^m$ satisfying 
$$|\{s_i : s_i\neq 0\}|=O\left(\frac{qn}{\varepsilon^2}\right),$$ such that
\[
(1-\eps)\cdot I\preceq \sum_{i=1}^m s_i v_i v_i^{\rot} \preceq (1+\eps)\cdot I.
\]
We will  use this algorithm to construct graph sparsifiers in Section~\ref{sec:analysis}.

\subsection{Overview of Our Approach}

\label{sec:BSS}

Our construction is based on a probabilistic view of the algorithm presented in Batson et al.~\cite{BSS12}. We  refer their algorithm \textsf{BSS} for short, and give a brief overview of the \textsf{BSS} algorithm at first.

  At a high level, the \textsf{BSS} algorithm proceeds by iterations, and adds a rank-1 matrix $c\cdot v_iv_i^{\rot}$ with some scaling factor $c$ to the currently constructed matrix $A_j$ in iteration $j$. To control the spectral properties of matrix $A_j$, the algorithm maintains two barrier values $u_j$ and $\ell_j$, and initially
     $u_0>0, \ell_0<0$. It was proven that one can always find a vector in $\{v_i\}_{i=1}^m$ and update $u_j,\ell_j$ in a proper manner in each iteration, such that the invariant
  \begin{equation}
\ell_j\mat I\prec\mat A_j\prec u_j\mat I\label{eq:app_inv}
\end{equation}
always holds, \cite{BSS12}. To guarantee this, Batson et al.~\cite{BSS12} introduces a potential function
\begin{equation}\label{eq:potenBSS}
\Phi_{u,\ell}(\mat A)\triangleq\tr(u\mat I-\mat A)^{-1}+\tr(\mat A-\ell\mat I)^{-1}
\end{equation}
to measure ``how far the eigenvalues of $A$ are from the barriers $u$ and $\ell$", since a small value of $\Phi_{u,\ell}(A)$ implies that no eigenvalue of $A$ is close to $u$ or $\ell$. With the help of the potential function, it was proven that, after $k=\Theta\left(n/\eps^2\right)$ iterations, it holds that $\ell_k\geq cu_k$ for some constant $c$, implying that the resulting matrix $A_k$ is a linear-sized and $A_k\approx_{O(\eps)} I$.

The original \textsf{BSS} algorithm is deterministic, and in each iteration the algorithm finds a rank-1 matrix which maximizes certain quantities. To informally explain our algorithm, let us look at the following randomized variant of the \textsf{BSS} algorithm: In each iteration, we choose a vector $v_i$ with probability $p_i$,  and  add a rank-1 matrix $$\Delta_A\triangleq \frac{\eps}{t}\cdot \frac{1}{p_i}\cdot v_iv_i^{\rot}$$ to the current matrix $A$. See Algorithm~\ref{algo:simple_BSS} for formal description.

\begin{algorithm}
\caption{Randomized \textsf{BSS} algorithm\label{algo:simple_BSS}}
\begin{algorithmic}[1]

\State $j=0$;

\State $\ell_{0}=-8n/\varepsilon$, $u_{0}=8n/\varepsilon$;

\State
$\mat A_{0}=\mat \mathbf{0}$;

\While{$u_j - \ell_j < 8n/\varepsilon$}

\State Let $t=\tr\left(u_{j}\mat I-\mat A_{j}\right)^{-1}+\tr\left(\mat A_{j}-\ell_{j}\mat I\right)^{-1}$;

\State Sample a vector $v_i$ with probability $p_{i}\triangleq\left(v_{i}^{\rot}\left(u_{j}\mat I-\mat A_{j}\right)^{-1}v_{i}+v_{i}^{\rot}\left(\mat A_{j}-\ell_{j}\mat I\right)^{-1}v_{i}\right)/t$;

\State $\ensuremath{\mat A_{j+1}=\mat A_{j}+\frac{\varepsilon}{t}\cdot \frac{1}{p_i}\cdot v_{i}v_{i}^{\rot}}$;

\State $\ensuremath{u_{j+1}= u_{j}+\frac{\varepsilon}{t\cdot (1-\varepsilon)}}$
and $\ell_{j+1}=\ell_{j}+\frac{\varepsilon}{t \cdot (1+\varepsilon)}$;

\State $j\leftarrow j+1$;

\EndWhile

\State \textbf{Return} $\mat A_j$;

\end{algorithmic}
\end{algorithm}

Let us look at any fixed iteration $j$, and analyze
how the added $\Delta_A$ impacts the potential function.
 We drop the subscript representing the iteration $j$ for simplicity.  After adding $\Delta_A$, the
  first-order approximation of $\Phi_{u,\ell}(\mat A)$
gives that
\begin{equation}
\Phi_{u,\ell}(\mat A+\Delta_{\mat A})\sim\Phi_{u,\ell}(\mat A)+ \left(u\mat I-\mat A\right)^{-2}\bullet\Delta_{\mat A} - \left(\mat A-\ell\mat I\right)^{-2}\bullet\Delta_{\mat A}.\label{eq:BSS_simple_first_order}
\end{equation}
Since
\begin{eqnarray*}
\mathbb{E}\left[\Delta_{\mat A}\right] & =\sum_{i=1}^mp_{i}\cdot \left(\frac{\varepsilon}{t}\cdot  \frac{1}{p_i}\cdot v_{i}v_{i}^{\rot}\right)
  =  \frac{\varepsilon}{t}\cdot \sum_{i=1}^mv_{i}v_{i}^{\rot}=\frac{\varepsilon}{t}\cdot  I,
\end{eqnarray*}
we have that
\begin{eqnarray*}
\mathbb{E}\left[\Phi_{u,\ell}(\mat A+\Delta_{\mat A})\right] & \sim & \Phi_{u,\ell}(\mat A)+
\frac{\varepsilon}{t}\cdot 
\left(u\mat I-\mat A\right)^{-2}\bullet I -\frac{\varepsilon}{t}\cdot  \left(\mat A-\ell\mat I\right)^{-2}\bullet I \\
 & = & \Phi_{u,\ell}(\mat A)+\frac{\varepsilon}{t}\cdot \tr\left(u\mat I-\mat A\right)^{-2}-\frac{\varepsilon}{t}\cdot \tr\left(\mat A-\ell\mat I\right)^{-2}\\
 & = & \Phi_{u,\ell}(\mat A)-\frac{\varepsilon}{t}\cdot \frac{\mathrm{d}}{\mathrm{d}u}\Phi_{u,\ell}(\mat A)-\frac{\varepsilon}{t}\cdot \frac{\mathrm{d}}{\mathrm{d}\ell}\Phi_{u,\ell}(\mat A).
\end{eqnarray*}
Notice that  if we increase $u$ by $\frac{\varepsilon}{t}$ and $\ell$
by $\frac{\varepsilon}{t}$, $\Phi_{u,\ell}$ approximately increases
by
\[
\frac{\varepsilon}{t}\cdot\frac{\mathrm{d}}{\mathrm{d}u}\Phi_{u,\ell}(\mat A)+\frac{\varepsilon}{t}\cdot\frac{\mathrm{d}}{\mathrm{d}\ell}\Phi_{u,\ell}(\mat A).
\]
Hence, comparing $\Phi_{u+\eps/t, \ell+\eps/t}(A+\Delta_A)$ with $\Phi_{u,\ell}(A)$, the increase of the potential function due to the change of
barrier values is approximately compensated by the drop of the potential function by the effect of  $\Delta_A$.
For a more rigorous analysis, we  need to  look at the higher-order terms
and increase $u$ slightly more than $\ell$ to compensate that. Batson
et al.~\cite{BSS12} gives the following estimate:

\begin{lem}[\cite{BSS12}, proof of Lemma 3.3 and 3.4]
\label{lem:BSSlemma}
Let $A\in\mathbb{R}^{n\times n}$, and $u,\ell$ be parameters satisfying
$\ell I\prec A\prec uI$. Suppose that $w\in\mathbb{R}^{n}$ satisfies
$ww^{\rot}\preceq\delta(uI-A)$ and $ww^{\rot}\preceq\delta(A-\ell I)$
for some $0<\delta<1$. Then, it holds that
\[
\Phi_{u,\ell}(A+ww^{\rot})\leq\Phi_{u,\ell}(A)+\frac{w^{\rot}(uI-A)^{-2}w}{1-\delta}
 -\frac{w^{\rot}(A-\ell I)^{-2} w}{1+\delta}.
\]

\end{lem}
The estimate above shows that the first-order approximation \eq{BSS_simple_first_order}  is good if $ww^{\rot}\preceq\delta(uI-A)$
and $ww^{\rot}\preceq\delta(A-\ell I)$ for small $\delta$. It is
easy to check that, by setting $\delta=\eps$, the added matrix $\Delta_A$ satisfies these two conditions, since \begin{eqnarray*}
\frac{\varepsilon}{t}\cdot \frac{1}{p_i}\cdot v_{i}v_{i}^{\rot} & = & \frac{\varepsilon\cdot v_{i}v_{i}^{\rot}}{v_{i}^{\rot}\left(u\mat I-\mat A\right)^{-1}v_{i}+v_{i}^{\rot}\left(\mat A-\ell\mat I\right)^{-1}v_{i}}\preceq  \frac{\varepsilon\cdot v_{i}v_{i}^{\rot}}{v_{i}^{\rot}\left(u\mat I-\mat A\right)^{-1}v_{i}} \preceq \varepsilon\left(u\mat I-\mat A\right),
\end{eqnarray*}
where we used the fact that $v v^\rot \preceq (v^\rot B^{-1} v) B$ for any vector $v$ and \textsf{PSD} matrix $B$. Similarly, we have that
\[
\frac{\eps}{t}\cdot\frac{1}{p_i}\cdot v_i v_i^{\rot}\preceq \eps(A-\ell I).
\]
Hence, if $\Phi_{u,\ell}(A)$ is small initially,  our crude calculations
above gives a good approximation and $\Phi_{u,\ell}(A)$ is small
throughout the executions of the  whole algorithm. Up to a  constant factor, this gives the same
result as \cite{BSS12}, and therefore Algorithm~\ref{algo:simple_BSS} constructs  an $\Theta(n/\varepsilon^{2})$-sized $(1+O(\varepsilon))$-spectral  sparsifier.

Our algorithm follows the same framework as  Algorithm~\ref{algo:simple_BSS}. However, to construct a spectral sparsifier in almost-linear time, we expect that the sampling probability $\{p_i\}_{i=1}^m$ of vectors (i) can be approximately computed fast, and (ii) can be further ``reused" for a few iterations.

For fast approximation of the sampling probabilities, we adopt the idea proposed  in \cite{zhu15}:  Instead of defining the potential function by \eq{potenBSS}, we define the potential function by
\[
\Phi_{u,\ell}(\mat{A}) \triangleq\tr(u\mat{I}-\mat{A})^{-q} + \tr(\mat{A}-\ell\mat{I})^{-q}.
\]
Since $q$ is a large constant, the value of the potential function becomes larger when some eigenvalue of $A$ is close to $u$ or $\ell$.
Hence, a bounded value of $\Phi_{u,\ell}(A)$ insures that the eigenvalues of $A$ never get too close to $u$ or $\ell$, which further allows us to compute the sampling probabilities $\{p_i\}_{i=1}^m$ efficiently simply by Taylor expansion.
 Moreover, by defining the potential function based on
 $\tr(\cdot)^{-q}$, one can prove a similar result
as Lemma \ref{lem:BSSlemma}. This
gives an alternative analysis of the algorithm presented in \cite{zhu15},
which is the first almost-quadratic time algorithm for constructing
linear-sized spectral sparsifiers.

To ``reuse" the sampling probabilities, we re-compute  $\{p_i\}_{i=1}^m$  after every  $\Theta\left(n^{1-1/q}\right)$ iterations: We show that as long as the sampling probability
satisfies
\[
p_{i}\geq C\cdot \frac{v_{i}^{\rot}\left(u\mat I-\mat A\right)^{-1}v_{i}+v_{i}^{\rot}\left(\mat A-\ell\mat I\right)^{-1}v_{i}}{\sum_{i=1}^m\left(v_{i}^{\rot}\left(u\mat I-\mat A\right)^{-1}v_{i}+v_{i}^{\rot}\left(\mat A-\ell\mat I\right)^{-1}v_{i}\right)}
\]
for some constant $C>0$, we can still sample $v_i$  with probability  $p_{i}$ and get the same guarantee on the potential function. The reason is as follows:
Assume that  $\Delta_{\mat A}=\sum_{i=1}^{T}\Delta_{\mat A,i}$  is the sum of the sampled matrices  within  $T=O\left(n^{1-1/q}\right)$  iterations. If a randomly chosen matrix $\Delta_{A,i}$ satisfies $\Delta_{\mat A,i}\preceq\frac{1}{Cq}\left(u\mat I-\mat A\right)$,
then by the matrix Chernoff bound $\Delta_{\mat A}\preceq\frac{1}{2}\left(u\mat I-\mat A\right)$ holds with
high probability.
By scaling every sampled rank-1 matrix $q$ times smaller, the sampling probability only changes by a constant factor within $T$ iterations.
Since we  choose $\Theta(n/\varepsilon^{2})$ vectors in total, our algorithm only recomputes the sampling probabilities $\Theta\left(n^{1/q}/\varepsilon^{2}\right)$ times. Hence, our algorithm runs in almost-linear time if $q$ is a large constant.

\subsection{Algorithm Description}

The algorithm follows the same framework as Algorithm~\ref{algo:simple_BSS}, and proceeds by iterations. Initially, the algorithm  sets  \[
u_0 \triangleq (2n)^{1/q}, \qquad \ell_0\triangleq  -(2n)^{1/q}, \qquad  A_0\triangleq\mathbf{0}.
\]
After  iteration $j$  the algorithm updates $u_j,\ell_j$ by $\Delta_{u,j}, \Delta_{\ell,j}$ respectively, i.e.,
\[
u_{j+1} \triangleq u_j+\Delta_{u,j}, \qquad \ell_{j+1} \triangleq \ell_j+\Delta_{\ell,j},
\]
and
updates  $A_j$ with respect to the chosen matrix in iteration $j$.  The  choice of $\Delta_{u,j}$ and $\Delta_{\ell,j}$ insures that
\[
\ell_j I\prec A_j\prec u_j I
\]
holds
for any $j$. In iteration $j$, the algorithm computes the \emph{relative effective resistance} of vectors $\{v_i\}_{i=1}^m$ defined by
\[
R_i\left(\mat{A}_j,u_j,\ell_j\right) \triangleq \vec{v}_i^{\rot} \left(u_j\mat{I}-\mat{A}_j\right)^{-1}\mat{v}_i + \vec{v}_i^{\rot}\left(\mat{A}_j-\ell_j \mat{I}\right)^{-1}\vec{v}_i,
\]
and samples $N_j$ vectors independently with replacement, where
 vector $v_i$ is chosen with probability proportional to $R_i(A_j,u_j,\ell_j)$, and
\[
N_j\triangleq \frac{1}{n^{2/q}} \left(\sum_{i=1}^m R_i(A_j,u_j,\ell_j) \right)
\min\left\{\lambda_{\min}(u_jI-A_j),\lambda_{\min}(A_j-\ell_j I)\right\}.
\]
The algorithm  sets $A_{j+1}$ to be the sum of $A_j$ and sampled $v_iv_i^{\rot}$ with proper reweighting.
For technical reasons, we define  $\Delta_{u,j}$ and $\Delta_{\ell,j}$ by
\[
\Delta_{u,j}\triangleq (1+2\varepsilon)\cdot\frac{\varepsilon\cdot N_j}{q\cdot \sum_{i=1}^m R_i(A_j,u_j,\ell_j)}, \qquad
\Delta_{\ell,j}\triangleq (1-2\varepsilon)\cdot\frac{\varepsilon\cdot N_j}{q\cdot \sum_{i=1}^m R_i(A_j,u_j,\ell_j)}.
\]
See Algorithm~\ref{aaa} for formal description.

\begin{algorithm}
\caption{Algorithm for constructing spectral sparsifiers\label{aaa}}
\begin{algorithmic}[1]
\Require $\eps \leq 1/120, q \geq 10$
\State $j=0$;
\State $ \ell_0=-(2n)^{1/q}, u_0=(2n)^{1/q}, A_0=\mathbf{0}$;
\While{$u_j-\ell_j <4\cdot(2n)^{1/q}$}
\State $W_j=\mathbf{0}$;
\State Compute $R_i(A_j, u_j,\ell_j)$ for all vectors $v_i$;
\State Sample $N_j$ vectors independently with replacement, where every $v_i$ is chosen with probability proportional to $R_i(A_j, u_j,\ell_j)$. For every sampled $v$, add $\eps/q\cdot (R_{i}(A_j,u_j,\ell_j ))^{-1}\cdot vv^{\rot}$ to $W_j$;
\State $A_{j+1}=A_j+W_j$;
\State $u_{j+1}= u_j+\Delta_{u,j}$, $\ell_{j+1}=\ell_j+\Delta_{\ell,j}$;
\State $j=j+1$;
\EndWhile
\State \textbf{Return} $A_j$;
\end{algorithmic}
\end{algorithm}

We remark that, although exact values of $N_j$ and relative effective resistances are difficult to compute in almost-linear time,
we can use approximated values of $R_i$ and $N_j$ instead. It is easy to see that  in each iteration an over estimate of $R_i$, and an under estimate of  $N_j$ with constant-factor approximation suffice for our purpose.

\section{Analysis}\label{sec:analysis}


We analyze Algorithm~\ref{aaa}
in this section.  To make the calculation less messy, we assume the following:

\begin{assumption}
We always assume that $0<\eps \leq 1/120$, and $q$ is an integer satisfying $q\geq 10$.
\end{assumption}

Our analysis is based on a potential function $\Phi_{u,\ell}$ with barrier values $u,\ell\in\mathbb{R}$.  Formally, for  a symmetric matrix $\mat{A}\in\mathbb{R}^{n\times n}$
with eigenvalues $\lambda_1\leq\cdots\leq\lambda_n$
and parameters $u,\ell$ satisfying $\ell I\prec A\prec uI$, let
\begin{align}
\Phi_{u,\ell}(\mat{A}) &\triangleq\tr(u\mat{I}-\mat{A})^{-q} + \tr(\mat{A}-\ell\mat{I})^{-q} \nonumber\\
&= \sum_{i=1}^n\left(\frac{1}{u-\lambda_i} \right)^q + \sum_{i=1}^n \left(\frac{1}{\lambda_i-\ell}\right)^q.\label{eq:defpot}
\end{align}

We will show how the potential function evolves after each iteration in \secref{approx}. Combing this with the ending condition of the algorithm, we will prove in \secref{aptanalysis} that the algorithm outputs a linear-sized spectral sparsifier. We will prove \thmref{graph} and \thmref{general} in \secref{mainproof}.

\subsection{Analysis of a Single Iteration\label{sec:approx}}

We analyze the sampling scheme within a single iteration, and drop the subscript representing the iteration $j$ for simplicity.  Recall that in each iteration the algorithm  samples $N$ vectors  independently from $\mathcal{V}=\{v_i\}_{i=1}^{m}$ satisfying $\sum_{i=1}^m v_iv_i^{\rot}=I$, where every vector $v_i$ is sampled with probability $\frac{R_i(A,u,\ell)}{\sum_{j=1}^{m} R_j(A,u,\ell)}$. We use $v_1,\cdots, v_{N}$ to denote these $N$ sampled vectors, and define the reweighted vectors by
\[
w_i\triangleq \sqrt{\frac{\eps}{q\cdot R_i(A,u,\ell)}}\cdot v_i,
\]
for any $1\leq i\leq N$. Let
\[
W\triangleq\sum_{i=1}^{N} w_iw_i^{\rot},
\]
and
we use $W\sim \mathcal{D}(A,u,\ell)$
to represent that $W$ is sampled in this way with parameters $A,u$ and $\ell$.
We will show that with high probability matrix $W$ satisfies $\mathbf{0}\preceq W\preceq \frac{1}{2}(uI-A)$.
We first recall the following Matrix Chernoff Bound.

\begin{lem}[Matrix Chernoff Bound, \cite{Tropp12}]\label{lem:matrixchernoff}
Let $\{\mat{X}_k\}$ be a finite sequence of independent, random, and self-adjoint matrices with dimension $n$. Assume that each random matrix satisfies
$
\mat{X}_k\succeq \mathbf{0}$, and $\lambda_{\max}(\mat{X}_k)\leq D$.
 Let
 $\mu\geq\lambda_{\max}\left(  \sum_k \Ex{\mat{X}_k} \right)$.
 Then, it holds for any $\delta\geq 0$ that
 \[
 \Pro{\lambda_{\max}\left( \sum_k \mat{X}_k \right) \geq (1+\delta)\mu} \leq n\cdot \left( \frac{\ce^{\delta}}{ (1+\delta)^{1+\delta}} \right)^{\mu/D}.
 \]
\end{lem}

\begin{lem}\label{lem:concentration}
Assume that the number of samples satisfies
 \[
N< \frac{2}{n^{2/q}}
\left( \sum_{i=1}^{m} R_i(A,u,\ell)\right)
\cdot  \lambda_{\min}(u\mat{I}-\mat{A}).
\]
Then, it holds that
\[
\Ex{W}=
\frac{\eps}{q}\cdot\frac{N}{\sum_{i=1}^{m} R_i(A,u,\ell)} \cdot\mat{I},
\]
and
\[
\Pro{ \mathbf{0}\preceq W\preceq \frac{1}{2} \cdot (u\mat{I}-\mat{A}) }
 \geq 1- \frac{\eps^2}{100 q n}.
\]
\end{lem}
\begin{proof}
By the description of the sampling procedure, it holds that
\[
\Ex{w_iw_i^{\rot}} = \sum_{j=1}^{m}  \frac{R_j(A,u,\ell)}{ \sum_{t=1}^{m} R_t(A,u,\ell)}  \cdot\frac{\eps}{q}\cdot  \frac{\vec{v}_j\vec{v}_j^{\rot}}{R_j(A,u,\ell)} =  \frac{\eps}{q} \cdot\frac{1}{\sum_{t=1}^{m} R_t(A,u,\ell)}\cdot\mat{I},
\]
and
\[
\Ex{W}=\Ex{\sum_{i=1}^{N} w_iw_i^{\rot}} = \frac{\eps}{q}\cdot\frac{N}{\sum_{i=1}^{m} R_i(A,u,\ell)} \cdot\mat{I},
\]
which proves the first statement.

Now for the second statement. Let $$\vec{z}_i = (u\mat{I}-\mat{A})^{-1/2} w_i.$$ It holds that
\begin{align*}
\tr\left(\vec{z}_i\vec{z}_i^{\rot} \right) & = \tr\left( (u\mat{I}-\mat{A})^{-1/2}w_iw_i^{\rot}(u\mat{I} -\mat{A})^{-1/2} \right) \\
& = \frac{\varepsilon}{q}\cdot \frac{ \tr\left( (u\mat{I}-\mat{A})^{-1/2} \vec{v}_i \vec{v}_i^{\rot} (u\mat{I}-\mat{A})^{-1/2} \right) }{R_i(A,u,\ell)} \\
& \leq \frac{\varepsilon}{q} \cdot\frac{\vec{v}_i^{\rot} (u\mat{I}-\mat{A})^{-1}\vec{v}_i }{\vec{v}_i^{\rot} (u\mat{I}-\mat{A})^{-1}\vec{v}_i+\vec{v}_i^{\rot} (\mat{A}-\ell\mat{I})^{-1}\vec{v}_i} \\
& \leq \frac{\varepsilon}{q},
\end{align*}
and $\lambda_{\max}(\vec{z}_i\vec{z}_i^{\rot})\leq \frac{\eps}{q}$.
Moreover, it holds that
\begin{align}
\Ex{\sum_{i=1}^{N}\vec{z}_i\vec{z}_i^{\rot} } &= \frac{\eps}{q}\cdot \frac{N}{\sum_{t=1}^{m} R_{t}(A,u,\ell)} \cdot (u\mat{I}-\mat{A})^{-1} \nonumber\\
&\preceq  \frac{\eps}{q} \cdot \frac{N}{\sum_{t=1}^{m} R_t(A,u,\ell)}\cdot\lambda_{\max}\left(  \frac{1}{u\mat{I}-\mat{A}}\right)\cdot\mat{I}.\label{eq:expz}
\end{align}
This implies that
\[
\lambda_{\max}\left(  \Ex{ \sum_{i=1}^{N} \vec{z}_i\vec{z}_i^{\rot} } \right) \leq \frac{\eps}{q} \cdot\frac{N}{\sum_{t=1}^{m} R_t(A,u,\ell)}\cdot\lambda_{\max}\left(  \frac{1}{u\mat{I}-\mat{A}}\right).
\]
By setting
\[
\mu=\frac{\eps}{q}\cdot\frac{N}{\sum_{i=1}^{m} R_i(A,u,\ell)}\cdot \lambda_{\max}\left(\frac{1}{u\mat{I}-\mat{A}}\right) ,
\]
it holds by the Matrix Chernoff Bound~(cf. \lemref{matrixchernoff}) that
\[
\Pro{\lambda_{\max}\left( \sum_{i=1}^{N} \vec{z}_i\vec{z}_i^{\rot} \right) \geq (1+\delta)\mu } \leq n\cdot \left( \frac{\ce^{\delta}}{ (1+\delta)^{1+\delta}} \right)^{\mu\cdot q/\eps}.
\]
Set the value of $1+\delta$ to be
\begin{align*}
1+\delta & = \frac{1}{2 \mu} = \frac{q}{2\eps  N}\cdot \left( \sum_{j=1}^{m} R_j(A,u,\ell) \right)\cdot \frac{1}{ \lambda_{\max}\left(\frac{1}{u\mat{I}-\mat{A}} \right) } \\
& = \frac{q}{2 \eps  N}\cdot \left( \sum_{j=1}^{m} R_j(A,u,\ell) \right)\cdot \lambda_{\min}(u\mat{I}-\mat{A}) \\
& \geq \frac{q}{4\eps}\cdot n^{2/q},
\end{align*}
where the last inequality follows from the condition on $N$.
Hence, with probability at least
\[
1-n\cdot \left( \frac{\ce^{\delta}}{ (1+\delta)^{1+\delta}} \right)^{\mu\cdot q/\eps} \geq1- n\cdot \left( \frac{\ce}{1+\delta} \right)^{(1+\delta)\cdot \mu\cdot q/\eps} \geq1-
n\ \left( \frac{\ce}{1+\delta}\right)^{\frac{q}{2\eps}} \geq1- \frac{\eps^2}{100 q n},
\]
we have that
\[
\lambda_{\max}\left( \sum_{i=1}^{N} z_iz_i^{\rot}\right) \leq (1+\delta)\cdot \mu =\frac{1}{2},
\]
which implies that
$ \mathbf{0}\preceq \sum_{i=1}^{N} z_iz_i^{\rot}\preceq \frac{1}{2}\cdot I$ and
$\mathbf{0}\preceq W\preceq \frac{1}{2}\cdot (uI-A)$.
\end{proof}

Now we analyze the change of the potential function after each iteration, and  show that the expected value of the potential function decreases over time.  By \lemref{concentration},  with probability at least $1- \frac{\eps^2}{100 q n}$, it holds that
\[
 \mathbf{0}\preceq W\preceq \frac{1}{2} (u\mat{I}-\mat{A}).
\]
We
define
\begin{align*}
\widetilde{\mathbb{E}}\left[f(W)\right] \triangleq \sum_{W\sim\mathcal{D}(A,u,\ell)} \Pro{W \text{ is chosen and } W\preceq \frac{1}{2} (uI-A)}\cdot f\left(W\right).
\end{align*}

\lemref{uplem1} below shows how the potential function changes after each iteration, and  plays a key role in our analysis.
This lemma was first proved in \cite{BSS12} for the case of $q=1$, and was extended in \cite{zhu15}  to general values of $q$.  For completeness, we include the proof of the lemma in the appendix.

\begin{lem}[\cite{zhu15}]\label{lem:uplem1}
Let $q\geq10$ and $\varepsilon\leq1/10$. Suppose that $w^{\rot} (uI-A)^{-1} w \leq \frac{\eps}{q}$ and  $w^{\rot} (A-\ell I)^{-1} w \leq \frac{\eps}{q}$.
It holds that
\[
\tr(A+ w w^{\rot} -\ell I)^{-q} \le \tr(A-\ell I)^{-q} - q(1-\varepsilon) \ w^\rot  (A-\ell I )^{-(q+1)} w,
\]
and
\[
\tr(uI-A-w w^\rot )^{-q} \le \tr(uI-A)^{-q} + q(1+\varepsilon)\ w^\rot  (uI-A )^{-(q+1)} w.
\]
\end{lem}

\begin{lem}\label{lem:potendecrease}
Let $j$ be any iteration. It holds that
\[
\TEx{\Phi_{u_{j+1}, \ell_{j+1}}(A_{j+1}) } \leq  \Phi_{u_j, \ell_j}(A_j).
\]
\end{lem}
\begin{proof}
Let $w_1 w_1^{\rot}, \cdots, w_{N_{j}}w_{N_{j}}^{\rot}$ be the matrices picked in iteration $j$, and define for any $0\leq i\leq N_{j}$ that
\[
B_i= A_j+ \sum_{t=1}^i w_t w_t^{\rot}.
\]
We study the change of the potential function after adding a rank-1 matrix within each iteration. For this reason, we use
\[
\overline{\Delta}_{u} =\frac{\Delta_{u,j}}{N_{j}} = (1+2\eps)\cdot \frac{\eps}{q\cdot\sum_{t=1}^m R_t(A_{j},u_{j},\ell_{j})},
\]
and
\[
\overline{\Delta}_{\ell} =\frac{\Delta_{\ell,j}}{N_{j}} = (1-2\eps)\cdot \frac{\eps}{q\cdot\sum_{t=1}^m R_t(A_{j},u_{j},\ell_{j})}
\]
to express the average change of the barrier values $\Delta_{u,j}$ and $\Delta_{\ell,j}$.
We further define for $0\leq j\leq N_j$ that
\[
\hat{u}_i=u_j+i\cdot \overline{\Delta}_{u}, \qquad \hat{\ell}_i=\ell_j+i\cdot \overline{\Delta}_{\ell}.
\]

Assuming $\mathbf{0}\preceq W_j\preceq \frac{1}{2} (u_{j}\mat{I}-\mat{A_j})$, we claim   that
\begin{equation}\label{eq:claim_potendecrease}
w_iw_i^{\rot}\preceq\frac{2\eps}{q}\cdot\left(\hat{u}_{i}I-B_{i-1}\right) \mathrm{~ and~ } w_iw_i^{\rot}\preceq\frac{2\eps}{q}\cdot\left(B_{i-1}-\hat{\ell}_{i} I\right),
\end{equation}
for any $1\leq i \leq N_j$.
Based on this, we  apply \lemref{uplem1} and get that
\begin{align}
\TEx{\Phi_{\hat{u}_{i},\hat{\ell}_{i}}\left(B_{i-1}+w_i w_i^\rot\right)} & \leq \Phi_{\hat{u}_{i},\hat{\ell}_{i}}(B_{i-1}) + q(1+2\varepsilon) \tr\left( (\hat{u}_{i}I-B_{i-1})^{-(q+1)} \Ex{w_i w_i^\rot} \right) \nonumber\\
& \qquad - q(1-2\varepsilon)\tr\left(  \left(B_{i-1}-\hat{\ell}_{i} I\right)^{-(q+1)}\Ex{w_i w_i^\rot} \right)\nonumber\\
& =  \Phi_{\hat{u}_{i},\hat{\ell}_{i}}(B_{i-1}) + q\cdot \overline{\Delta}_{u} \cdot\tr\left( (\hat{u}_{i}I-B_{i-1})^{-(q+1)} \right)\nonumber\\
& \qquad  - q\cdot\overline{\Delta}_{\ell} \cdot \tr\left(  (B_{i-1}-\hat{\ell}_{i} I )^{-(q+1)}\right).\label{eq:update1}
\end{align}
We define a function $f_i$ by
\[
f_i(t)=\tr\left( \left(\hat{u}_{i-1}+t\cdot \overline{\Delta}_{u} \right)I - B_{i-1}\right)^{-q} + \tr\left( B_{i-1}- \left(\hat{\ell}_{i-1}+t\cdot \overline{\Delta}_{\ell}\right)I \right)^{-q}.
\]
Notice that
\begin{align*}
\frac{\mathrm{d}f_{i}(t)}{\mathrm{d}t} & = -q\cdot \overline{\Delta}_{u} \cdot \tr\left( \left(\hat{u}_{i-1}+t\cdot \overline{\Delta}_{u}\right)I - B_{i-1}\right)^{-(q+1) }+ q\cdot \overline{\Delta}_{\ell}\cdot  \tr\left( B_{i-1}- \left(\hat{\ell}_{i-1}+t\cdot \overline{\Delta}_{\ell}\right)I \right)^{-(q+1)}.
\end{align*}
Since $f$ is convex, we have that
\begin{equation}\label{eq:der2}
\frac{\mathrm{d}f_{i}(t)}{\mathrm{d}t} \Big|_{t=1} \geq f_i(1)-f_i(0) =\Phi_{\hat{u}_{i}, \hat{\ell}_{i}}(B_{i-1}) - \Phi_{\hat{u}_{i-1}, \hat{\ell}_{i-1}}(B_{i-1}).
\end{equation}
Putting \eq{update1} and \eq{der2} together,
 we have that
\[
\TEx{\Phi_{\hat{u}_{i},\hat{\ell}_{i}}(B_{i})} \leq \Phi_{\hat{u}_{i},\hat{\ell}_{i}}(B_{i-1}) -\frac{\mathrm{d}f_{i}(t)}{\mathrm{d}t} \Big|_{t=1}\leq  \Phi_{\hat{u}_{i-1}, \hat{\ell}_{i-1}}( B_{i-1}).
\]
Repeat this argument, we have that
\[
\TEx{\Phi_{u_{j+1}, \ell_{j+1}}(A_{j+1})} = \TEx{\Phi_{\hat{u}_{N_j},\hat{\ell}_{N_j}}\left(B_{N_j}\right)} \leq \Phi_{\hat{u}_{0}, \hat{\ell}_{0}}(B_{0}) = \Phi_{u_{j}, \ell_{j}}(A_{j}),
\]
which proves the statement.

So, it suffices to prove the claim \eq{claim_potendecrease}. Since $v v^\rot \preceq (v^\rot B^{-1} v) B$ for any vector $v$ and \textsf{PSD} matrix $B$, we have that
$$ \frac{v_i v_i^\rot}{R_i(A_{j},u_{j},\ell_{j})} \preceq \frac{v_i v_i^\rot}{v_i^\rot (u_j I-A_j)^{-1} v_i} \preceq
u_j I - A_j.
$$
By the assumption of $W_{j}\preceq\frac{1}{2}(u_{j}I-A_{j})$, it holds that 
$$ w_i w_i^\rot  = \frac{\eps}{q R_i(A_{j},u_{j},\ell_{j})}v_{i}v_{i}^\rot
\preceq \frac{\varepsilon}{q}\left(u_{j}I-A_{j}\right)
\preceq \frac{2\varepsilon}{q}\left(\hat{u}_{i}I-B_{i-1}\right).$$
This proves the first  statement of the claim.

For the second statement, notice that
$$
\ell_{j+1}-\ell_{j}  =  \Delta_{\ell,j}
\leq  \frac{\varepsilon N_j }{q \sum_{t=1}^m R_t(A_{j},u_{j},\ell_{j})}
 \leq  \frac{1}{2}\cdot\lambda_{\min}(A_{j}-\ell_{j}I)
$$
and hence
$$w_{i}w_{i}^{\rot}
 \preceq  \frac{\varepsilon}{q}\left(A_{j}-\ell_{j}I\right)
 \preceq \frac{2\varepsilon}{q}\left(A_{j}-\hat{\ell}_{i}I\right)
 \preceq \frac{2\varepsilon}{q}\left(B_{i-1}-\hat{\ell}_{i}I\right).
$$
\end{proof}

\subsection{Analysis of the Approximation Guarantee\label{sec:aptanalysis}}

In this subsection we will prove that the algorithm produces a linear-sized $(1+O(\eps))$-spectral sparsifier. 
We assume that the algorithm finishes after $k$ iterations, and will prove that the output $A_k$ is  a $(1+O(\eps))$-spectral sparsifier. It  suffices to show that  the condition number of $A_k$
is small, which follows directly from our setting of parameters.

\begin{lem}\label{lem:condition}
The output matrix $A_k$ has condition number at most $1+O(\eps)$.
\end{lem}
\begin{proof}
Since the condition number of $A_{k}$ is at most
$$\frac{u_{k}}{\ell_{k}}=\left(1-\frac{u_{k}-\ell_{k}}{u_{k}}\right)^{-1},$$
it suffices to prove that $(u_{k}-\ell_{k})/u_{k}=O(\eps)$.

Since the increase rate of $\Delta_{u,j}-\Delta_{\ell,j}$ with respect to $\Delta_{u,j}$ for any iteration $j$ is
\[
\frac{\Delta_{u,j}- \Delta_{\ell,j}}{\Delta_{u,j}} =\frac{(1+2\eps ) - (1-2\eps)}{1+2\eps} = \frac{4\eps}{1+2\eps}\leq 4\eps,
\]
we have that
\begin{align*}
\frac{u_{k}-\ell_{k}}{u_{k}}
& = \frac{2\cdot (2n)^{1/q}+\sum_{j=0}^{k-1}\left(\Delta_{u,j}-\Delta_{\ell,j}\right)}{(2n)^{1/q}+\sum_{j=0}^{k-1}\Delta_{u,j}}\\
&
\leq \frac{2\cdot (2n)^{1/q}+\sum_{j=0}^{k-1}\left(\Delta_{u,j}-\Delta_{\ell,j}\right)}{(2n)^{1/q}+(4\eps)^{-1}\sum_{j=0}^{k-1}\left(\Delta_{u,j}-\Delta_{\ell,j}\right)}.
\end{align*}
By the ending condition of the algorithm, it holds that $u_k-\ell_k\geq 4\cdot(2n)^{1/q}$, i.e.
\[
\sum_{j=0}^{k-1}\left( \Delta_{u,j} -\Delta_{\ell,j} \right) \geq 2\cdot (2n)^{1/q}.
\]
Hence, it holds that
$$\frac{u_{k}-\ell_{k}}{u_{k}}\leq\frac{2\cdot(2n)^{1/q}+2\cdot(2n)^{1/q}}{(2n)^{1/q}+\left(4\eps\right)^{-1}2\cdot(2n)^{1/q}}\leq8\eps,$$
which finishes the proof.
\end{proof}

Now we  prove that the algorithm
finishes in $O\left(\frac{qn^{3/q}}{\eps^2}\right)$ iterations, and
picks  $O\left(\frac{qn}{\eps^2}\right)$ vectors in total.

\begin{lem}\label{lem:iterationneeded}
The following statements hold:
\begin{itemize}
\item With probability at least $4/5$, the algorithm finishes in $\frac{10q n^{3/q}}{\eps^2}$ iterations.
\item With probability at least $4/5$, the algorithm chooses at most $\frac{10qn}{\eps^2}$ vectors.
\end{itemize}
\label{lem:bounditerations}
\end{lem}
\begin{proof}
Notice that after iteration $j$
the barrier gap $u_j - \ell_j$ is increased by
\begin{align*}
\Delta_{u, j } - \Delta_{\ell, j } & = \frac{4\eps^2}{q}\ \frac{N_j}{ \sum_{i=1}^m R_i(A_j, u_j, \ell_j)} \\
& = \frac{4\eps^2}{q}\  \frac{1}{n^{2/q}}\cdot
\min\left\{\lambda_{\min}(u_{j}I-A_j),\lambda_{\min}(A_j-\ell_j I) \right\}\\
& \geq\frac{4\eps^2}{q}\  \frac{1}{n^{2/q}}\cdot \left( \Phi_{u_j, \ell_j}(A_j)\right)^{-1/q}.
\end{align*}
Since the algorithm finishes within  $k$ iterations if \[\sum_{j=0}^{k-1} (\Delta_{u,j} -\Delta_{\ell,j}) \geq 2\cdot(2n)^{1/q},\] it holds that
\begin{align*}
\Pro{\mbox{algorithm finishes within $k$ iterations}} & \geq \Pro{ \sum_{j=0}^{k-1} (\Delta_{u,j} -\Delta_{\ell,j}) \geq 2\cdot(2n)^{1/q}  } \\
& \geq \Pro{ \sum_{j=0}^{k-1}  \frac{4\eps^2}{q n^{2/q}}\cdot \left( \Phi_{u_j, \ell_j}(A_j)\right)^{-1/q} \geq 2\cdot (2n)^{1/q}  }  \\
& = \Pro{ \sum_{j=0}^{k-1} \left( \Phi_{u_j, \ell_j}(A_j)\right)^{-1/q} \geq  \frac{q}{2\eps^2}\cdot \left( 2n^3 \right)^{1/q} } \\
& \geq \Pro{\sum_{j=0}^{k-1} \left( \Phi_{u_j, \ell_j}(A_j)\right)^{1/q} \leq 2\cdot \frac{k^2\eps^2}{q}\cdot \left( \frac{1}{2 n^3} \right)^{1/q} },
\end{align*}
where the last inequality follows from the fact that
\[
\left(\sum_{j=0}^{k-1} \left( \Phi_{u_j, \ell_j}(A_j)\right)^{-1/q} \right)
\cdot
\left(\sum_{j=0}^{k-1} \left( \Phi_{u_j, \ell_j}(A_j)\right)^{1/q} \right)
\geq k^2.
\]
By \lemref{concentration},  every picked matrix $W_j$ in iteration $j$ satisfies \[0\preceq W_j\preceq \frac{1}{2} \cdot \left(u_jI-A\right)\]with probability at least  $1-\frac{\eps^2}{100 q n}$, and with probability $9/10$ all matrices picked in $k=\frac{10qn}{\eps^2}$ iterations satisfy the condition above.
Also,  by \lemref{potendecrease} we have that
\begin{equation}\label{eq:expchain}
\TEx{\sum_{j=0}^{k-1} (\Phi_{u_j, \ell_j}(A_j))^{1/q} }=\sum_{j=0}^{k-1} \TEx{(\Phi_{u_j, \ell_j}(A_j))^{1/q} } \leq \sum_{j=0}^{k-1} \left( \TEx{\Phi_{u_j, \ell_j}(A_j)}\right)^{1/q} \leq k,
\end{equation}
since the initial value of the potential function is at most 1. Therefore,  it holds that
\begin{align*}
\lefteqn{\Pro{\mbox{algorithm finishes in more than $k$ iterations}}}\\
& \leq \Pro{\sum_{j=0}^{k-1} \left( \Phi_{u_j, \ell_j}(A_j)\right)^{1/q} \geq 2\cdot \frac{k^2\eps^2}{q}\cdot \left( \frac{1}{2 n^3} \right)^{1/q} } \\
& \leq \Pro{\sum_{j=0}^{k-1} \left( \Phi_{u_j, \ell_j}(A_j)\right)^{1/q} \geq 2\cdot \frac{k^2\eps^2}{q}\cdot \left( \frac{1}{2 n^3} \right)^{1/q} \text{ and }  \forall j:W_j\preceq \frac{1}{2} (u_jI-A_j)} \\
& \qquad + \Pro{\exists j: W_j\not\preceq \frac{1}{2} (u_jI-A_j)} \\
& \leq \frac{q}{2\cdot k\eps^2}\cdot \left(2n^3\right)^{1/q} + 1/10 \leq 1/5,
\end{align*}
where the  second last inequity follows from  Markov's inequality and
\eq{expchain}, and the last inequality follows by our choice of $k$. This proves the first statement.

Now for the second statement.
Notice that for every vector chosen  in iteration $j$, the barrier gap $\Delta_{u,j}-\Delta_{\ell,j}$ is increased on average by
\[
\frac{\Delta_{u,j}-\Delta_{\ell,j}}{N_j} = \frac{4\varepsilon^2}{q\sum_{i=1}^m R_i(A_j, u_j,\ell_j)}.
\]
To bound $R_i(A_j, u_j,\ell_j)$, let the eigenvalues of matrix $\mat{A_j}$ be $\lambda_1,\cdots,\lambda_n$. Then, it holds that
\begin{align*}
\sum_{i=1}^{m} R_i(\mat{A_j},u_j,\ell_j) & =  \sum_{i=1}^{m} \vec{v}_i^{\rot} (u_j\mat{I}-\mat{A_j})^{-1}\mat{v}_i + \sum_{i=1}^{m} \vec{v}_i^{\rot}(\mat{A_j}-\ell_j \mat{I})^{-1}\vec{v}_i \\
& = \sum_{i=1}^n \frac{1}{u_j-\lambda_i} + \sum_{i=1}^n \frac{1}{\lambda_i- \ell_j}\\
 & \leq \left( \sum_{i=1}^n (u_j-\lambda_i)^{-q} + \sum_{i=1}^n (\lambda_i-\ell_j)^{-q} \right)^{1/q} (2n)^{1-1/q}\\
& = \left( \Phi_{u_j,\ell_j}(\mat{A_j}) \right)^{1/q}\cdot (2n)^{1-1/q}.
\end{align*}
Therefore, we have that
\begin{align}
\frac{\Delta_{u,j}-\Delta_{\ell,j}}{N_j} \geq \frac{4\eps^2}{q}\cdot\frac{1}{(2n)^{1-1/q}\cdot (\Phi_{u_j, \ell_j}(A_j))^{1/q}}.\label{eq:lbpot}
\end{align}

Let $v_1,\cdots, v_z$ be the vectors sampled by the algorithm, and
$v_j$ is picked in iteration $\tau_j$,  where $1\leq j\leq z$.
We first assume that the algorithm could check the ending condition after adding every single vector. In such  case, it holds that
\begin{align*}
\lefteqn{\Pro{\mbox{algorithm finishes after choosing $z$ vectors}}}\\
 & \geq \Pro{\sum_{j=1}^z \frac{4\eps^2}{q}\cdot\frac{1}{(2n)^{1-1/q} \cdot(\Phi_{u_{\tau_j}, \ell_{\tau_j}}(A_{\tau_j}))^{1/q}}\geq 2\cdot(2n)^{1/q}} \\
& = \Pro{\sum_{j=1}^z (\Phi_{u_{\tau_j}, \ell_{\tau_j} } (A_{\tau_j}))^{-1/q} \geq   qn/\eps^2 }.
\end{align*}
Following the same proof as the first part and noticing that in the final iteration the algorithm chooses at most $O(n)$ extra vectors, we obtain the second statement.
\end{proof}

\subsection{Proof of the Main Results\label{sec:mainproof}}

Now we analyze the runtime of the algorithm, and prove the main results. 
We first analyze the algorithm for sparsifying sums of rank-1 \textsf{PSD} matrices, and prove \thmref{general}.

\begin{proof}[Proof of \thmref{general}]
By \lemref{iterationneeded}, with probability at least $4/5$ the algorithm chooses at most $\frac{10qn}{\varepsilon^2}$ vectors,
 and by \lemref{condition} the condition number  of $A_k$ is at most $1+O(\eps)$,  implying that the  matrix $A_k$ is a $(1+O(\eps))$-approximation of $I$. These two results together prove that $A_k$ is a linear-sized spectral sparsifier.

For the runtime, \lemref{iterationneeded} proves that the algorithm finishes  in  $\frac{10q n^{3/q}}{\eps^2} $ iterations,
and it is easy to see that  all the required quantities in each iteration can be approximately computed in  $\tilde{O}(m\cdot n^{\omega-1})$ time using fast matrix multiplication. Therefore, the total runtime of the algorithm is  $\tilde{O}\left(\frac{q \cdot m }{\eps^2}\cdot n^{\omega-1+3/q}\right)$.
\end{proof}

Next we show how to apply our algorithm in the graph setting, and prove \thmref{graph}.
Let $L=\sum_{i=1}^m u_iu_i^{\rot}$ be the Laplacian matrix of an undirected graph $G$, where $u_iu_i^{\rot}$ is the Laplacian matrix of the graph consisting of  a single edge $e_i$. 
By setting
\[
v_i=L^{-1/2}u_i
\]
for $1\leq i\leq m$, it is easy to see that constructing a spectral sparsifier of $G$ is equivalent to sparsifing the matrix $\sum_{i=1}^m v_iv_i^{\rot}$.
We will present in the appendix almost-linear time algorithms to approximate the required quantities \[ \lambda_{\min}\left(u_jI-A_j\right),  \lambda_{\min}\left(A_j-\ell_j I\right), v_i^{\rot}\left(u_jI-A_j\right)^{-1}v_i, \mbox{\ and\ } v_i^{\rot}\left(A_j-\ell_j I\right)^{-1} v_i\] in each iteration, and this gives  \thmref{graph}.
\begin{proof}[Proof of \thmref{graph}]
By applying the same analysis as in the proof of \thmref{general}, we know that the output matrix $A_k$ is  a linear-sized spectral sparsifier, and it suffices to analyze the runtime of the algorithm.

  By
\lemref{concentration} and the Union Bound, with probability at least $9/10$ all the matrices picked in $k=\frac{10q n^{3/q}}{\eps^2}$ iterations  satisfy \[W_j\preceq \frac{1}{2} (u_jI-A_j).\] Conditioning on the event, with constant probability  $\Ex{\Phi_{u_j, \ell_j}(A_j)}\leq2$
for all iterations $j$, and by Markov's inequality with high probability it holds that $\Phi_{u_j,\ell_j}(A_j)=O\left(\frac{qn}{
\eps^2}\right)$ for all iterations $j$.

On the other hand, notice that 
it holds for any $1\leq j\leq n$ that
\[
(u-\lambda_j)^{-q} \leq \sum_{i=1}^n (u-\lambda_i)^{-q} < \Phi_{u,\ell}(\mat{A}),\]
which implies that $\lambda_j< u- \left( \Phi_{u,\ell}(\mat{A}) \right)^{-1/q}$. Similarly, it holds that $\lambda_j> \ell+\left( \Phi_{u,\ell}(\mat{A})  \right)^{-1/q}$ for any $1\leq j\leq n$.
Therefore, we have that
\[
\left( \ell_j+ O\left(  \left(\frac{\eps^2}{qn} \right)^{1/q} \right) \right)\mat{I}\prec \mat{A_j}\prec \left( u_j- O\left(  \left(\frac{\eps^2}{qn} \right)^{1/q} \right)\right)\mat{I}.
\]
Since both of $u_j$ and $\ell_j$ are of the order $O(n^{1/q})$, we set  $\eta=O\left( (\eps/n)^{2/q} \right)$ and obtain that 
\[
(\ell_j + |\ell_j| \eta) I\prec A_j\prec (1-\eta)u_jI.
\]
Hence, we apply   \lemref{procedure2} and \lemref{procedure3} to  compute all required quantities in each iteration up to constant approximation in time
\[
\tilde{O}\left( \frac{m}{\eps^2\cdot\eta}\right)= \tilde{O}\left(\frac{m\cdot n^{2/q}}{\eps^{2+2/q}} \right).
\]
Since by \lemref{iterationneeded}  the algorithm finishes  in  $\frac{10q n^{3/q}}{\eps^2}$ iterations with probability at least $4/5$, the total runtime of the algorithm is
\[
\tilde{O}\left(\frac{q\cdot m\cdot n^{5/q}}{\eps^{4+4/q}}\right).
\]
\end{proof}


\section*{Acknowledgment}
This work was partially supported by NSF awards 0843915 and 1111109. Part of this work was done
while both authors were visiting the Simons Institute for
the Theory of Computing, UC Berkeley, and the second author was affiliated with the Max Planck Institute for Informatics, Germany.
We thank Zeyuan Allen-Zhu, Zhenyu Liao, and Lorenzo Orecchia for sending us their manuscript of \cite{zhu15}
and the inspiring talk Zeyuan Allen-Zhu gave at the Simons Institute for the Theory of Computing.
 Finally, we  thank Michael Cohen for pointing out a gap in a previous version of the paper and his fixes for
the gap, as well as Lap-Chi Lau for many insightful comments on improving the presentation of the paper.



%

\bibliographystyle{plain}

\bibliography{reference}

\section{Omitted Proofs}

\subsection{Estimates of the Potential Functions}

In this subsection we prove \lemref{uplem1}. We first list the following two lemmas, which will be used in our proof.
\begin{lem}[Sherman-Morrison Formula]\label{lem:woodbury}
Let $\mat{A}\in \mathbb{R}^{n\times n}$ be an invertible matrix, and $u,v\in\mathbb{R}^n$. Suppose that $1+v^{\rot}A^{-1}u\neq 0$. Then it holds that
\[
(\mat{A}+u v^\rot)^{-1} = \mat{A}^{-1} - \frac{\mat{A}^{-1} u v^\rot\mat{A}^{-1}}{1+v^\rot \mat{A}^{-1} u}.
\]
\end{lem}

\begin{lem}[Lieb Thirring Inequality, \cite{lieb1976inequalities}]\label{lem:lieb_inq}
Let $\mat{A}$ and $\mat{B}$ be positive definite matrices, and $q\geq1$. Then it holds that
\[
\tr(B A B)^q \leq \tr(B^q A^q B^q).
\]
\end{lem}

\begin{proof}[Proof of \lemref{uplem1}]
Let $Y=A-\ell I$. By the Sherman-Morrison Formula~(\lemref{woodbury}), it holds that
\begin{equation}\label{eq:woodburyinstance}
\tr(\mat{Y}+ w w^{\rot})^{-q}=\tr\left( \mat{Y}^{-1} - \frac{\mat{Y}^{-1} w w^{\rot}\mat{Y}^{-1}}{1 + w^{\rot}\mat{Y}^{-1} w} \right)^q.
\end{equation}
By the assumption of  $w^{\rot}\mat{Y}^{-1} w \leq \frac{\eps}{q}$, we have that
\begin{align}
\tr(\mat{Y}+w w^{\rot})^{-q} & \leq \tr\left( \mat{Y}^{-1} - \frac{\mat{Y}^{-1} w w^{\rot} \mat{Y}^{-1}}{1+\varepsilon / q} \right)^{q}\label{eq:ineq11}\\
& = \tr\left( \mat{Y}^{-1/2}\left( \mat{I} - \frac{\mat{Y}^{-1/2} w w^{\rot}\mat{Y}^{-1/2} }{1+\varepsilon/q} \right)\mat{Y}^{-1/2} \right)^q\nonumber\\
& \leq \tr\left( \mat{Y}^{-q/2}\left(\mat{I}- \frac{\mat{Y}^{-1/2} w w^{\rot}\mat{Y}^{-1/2} }{1+\varepsilon/q} \right)^q\mat{Y}^{-q/2} \right) \label{eq:ineq21}\\
& = \tr \left( \mat{Y}^{-q}\left( \mat{I}-\frac{\mat{Y}^{-1/2} w w^{\rot}\mat{Y}^{-1/2} }{1+\varepsilon/q} \right)^q  \right), \label{eq:eq21}
\end{align}
where \eq{ineq11} uses the fact that $\mat{A}\preceq \mat{B}$ implies that $\tr\left(\mat{A}^q\right)\leq \tr \left( \mat{B}^q	 \right)$, \eq{ineq21} follows from the Lieb-Thirring inequality (\lemref{lieb_inq}), and \eq{eq21} uses the fact that the trace is invariant under cyclic permutations.

Let 
\[
\mat{D}=\frac{\mat{Y} ^{-1/2} w w^\rot \mat{Y}^{-1/2}}{1+\varepsilon/q}.
\]
 Note that $0 \preceq D\preceq\frac{\varepsilon}{q}\cdot \mat{I}$, and
\begin{align*}
(I-D)^{q} & \preceq  I-qD+\frac{q(q-1)}{2}D^{2}\\
 & \preceq  I-\left(q-\frac{\varepsilon(q-1)}{2}\right)D
\end{align*}
Therefore, we have that
\begin{align*}
\left(I-\frac{\mat{Y}^{-1/2} w w^\rot \mat{Y}^{-1/2}}{1+\varepsilon/q}\right)^{q}
& \preceq I-\left(q-\frac{\varepsilon(q-1)}{2}\right)\frac{\mat{Y}^{-1/2} w w^\rot \mat{Y}^{-1/2}}{1+\varepsilon/q}\\
& \preceq I-\left(q-\frac{\varepsilon(q-1)}{2}\right)\left(1-\frac{\varepsilon}{q}\right)\mat{Y}^{-1/2} w w^\rot \mat{Y}^{-1/2}\\
& \preceq I-q\left(1-\frac{\varepsilon(q+1)}{2q}\right)\mat{Y}^{-1/2} w w^\rot \mat{Y}^{-1/2}\\
& \preceq I-q\left(1-\varepsilon\right) \mat{Y}^{-1/2} w w^\rot \mat{Y}^{-1/2}.
\end{align*}
This implies that
\[
\tr(Y + w w^{\rot})^{-q} \le
\tr\left( Y^{-q} \left( I-q(1-\eps) Y^{-1/2} ww^{\rot} Y^{-1/2} \right)   \right)
\leq \tr\left(Y^{-q}\right) - q(1- \varepsilon) \ w^\rot  Y^{-(q+1)} w,
\]
which proves the first statement.

Now for the second inequality.  Let $Z=uI-A$.
By the Sherman-Morrison Formula~(\lemref{woodbury}), it holds that
\begin{equation}\label{eq:woodburyinstance}
\tr(\mat{Z}- w w^{\rot})^{-q}=\tr\left( \mat{Z}^{-1} + \frac{\mat{Z}^{-1} w w^{\rot}\mat{Z}^{-1}}{1 - w^{\rot}\mat{Z}^{-1} w} \right)^q.
\end{equation}
By the assumption of  $w^{\rot}\mat{Z}^{-1} w \leq \frac{\eps}{q}$, it holds that
\begin{align}
\tr(\mat{Z}-w w^{\rot})^{-q} & \leq \tr\left( \mat{Z}^{-1} + \frac{\mat{Z}^{-1} w w^{\rot} \mat{Z}^{-1}}{1-\varepsilon / q} \right)^{q}\label{eq:ineq1}\\
& = \tr\left( \mat{Z}^{-1/2}\left( \mat{I} + \frac{\mat{Z}^{-1/2} w w^{\rot}\mat{Z}^{-1/2} }{1-\varepsilon/q} \right)\mat{Z}^{-1/2} \right)^q\nonumber\\
& \leq \tr\left( \mat{Z}^{-q/2}\left(\mat{I} + \frac{\mat{Z}^{-1/2} w w^{\rot}\mat{Z}^{-1/2} }{1-\varepsilon/q} \right)^q\mat{Z}^{-q/2} \right) \label{eq:ineq2}\\
& = \tr \left( \mat{Z}^{-q}\left( \mat{I} + \frac{\mat{Z}^{-1/2} w w^{\rot}\mat{Z}^{-1/2} }{1-\varepsilon/q} \right)^q  \right), \label{eq:eq2}
\end{align}
where \eq{ineq1} uses the fact that $\mat{A}\preceq \mat{B}$ implies that $\tr\left(\mat{A}^q\right)\leq \tr \left( \mat{B}^q	 \right)$, \eq{ineq2} follows from the Lieb-Thirring inequality (\lemref{lieb_inq}), and \eq{eq2} uses the fact that the trace is invariant under cyclic permutations.

Let \[\mat{E}=\mat{Z} ^{-1/2} w w^\rot \mat{Z}^{-1/2}.\]
Combing $E \preceq \frac{\varepsilon}{q} \cdot\mat{I}$ with the assumption that $q\geq10$ and $\varepsilon\leq1/10$, we
have that
\begin{align*}
\left(I+\frac{E}{1-\varepsilon/q}\right)^{q}
& \preceq I+ \frac{qE}{1-\varepsilon/q}+\frac{q(q-1)}{2}\left(1+\frac{\varepsilon/q}{1-\varepsilon/q}\right)^{q-2}\left(\frac{E}{1-\varepsilon/q}\right)^{2}\\
& \preceq I+q\left(1+1.1\frac{\varepsilon}{q}\right)E+1.4\frac{q(q-1)}{2}E^{2}\\
& \preceq I+q\left(1+0.3\varepsilon\right)E+0.7\varepsilon qE\\
& \preceq I+q\left(1+\varepsilon\right)E.
\end{align*}
Therefore, we have that
\[
\tr(Z - w w^{\rot})^{-q} \le \tr\left(Z^{-q}\right) + q(1+ \varepsilon) \ w^\rot  Z^{-(q+1)} w,
\]
which proves the second statement.
\end{proof}

\subsection{Implementation of the Algorithm}

In this section, we show that the algorithm for constructing graph sparsification runs in almost-linear time. Based on previous discussion, we only need to prove that, for any iteration $j$, the number
 of samples $N_j$ and $\{R_i(A_j,u_j,\ell_j)\}_{i=1}^m$ can be approximately computed in almost-linear time.
 By definition, it suffices to compute $\lambda_{\min}\left(u_jI-A_j\right)$, $\lambda_{\min}\left(A_j-\ell_j I\right)$, $v_i^{\rot}\left(u_jI-A_j\right)^{-1}v_i$, and $v_i^{\rot}\left(A_j-\ell_j I\right)^{-1} v_i$ for all $i$.  For simplicity we drop the subscript $j$ expressing the iterations in this subsection.
We will assume that  the following assumption holds on $A$. We remark that an almost-linear time algorithm for computing similar quantities was shown in \cite{zhu15}.

\begin{assumption}\label{assump1}
Let $L$ and $\tilde{L}$ be the Laplacian matrices of graph $G$ and its subgraph after reweighting.
Let $A=L^{-1/2}\tilde{L}L^{-1/2}$, and assume  that
\[
(\ell + |\ell|\eta)\cdot I \prec A\prec (1-\eta)u\cdot I
\]
holds
for some $0<\eta<1$.
\end{assumption}

\begin{lem}\label{lem:procedure1}
Under Assumption~\ref{assump1}, the following statements hold:
\begin{itemize}
\item We can construct a matrix $S_u$ such that
\[
S_u\approx_{\eps/10} (uI-A)^{-1/2},
\]
and $S_u=p(A)$ for a polynomial $p$ of degree $O\left(\frac{\log (1/\eps\eta)}{\eta}\right)$.
\item We can construct a matrix $S_{\ell}$ such that
 \[
S_{\ell} \approx_{\eps/10} (A-\ell I)^{-1/2}.
 \]
 Moreover,  $S_{\ell}$ is of the form $(A')^{-1/2} q((A')^{-1})$,where $q$ is a polynomial of degree  $O\left(\frac{\log (1/\eps\eta)}{\eta}\right)$ and $A'=L^{-1/2} L' L^{-1/2}$ for  some Laplacian matrix $L'$.
\end{itemize}
\end{lem}
\begin{proof}
By Taylor expansion, it holds that
\[
(1-x)^{-1/2}=1+\sum_{k=1}^{\infty} \prod_{j=0}^{k-1} \left( j+\frac{1}{2} \right)\frac{x^k}{k!}.
\]
We define for any $T\in\mathbb{N}$ that
\[
p_T(x)=1+\sum_{k=1}^T \prod_{j=0}^{k-1} \left( j+\frac{1}{2} \right)\frac{x^k}{k!}.
\]
Then, it holds for any $0<x<1-\eta$ that
\begin{align*}
p_T(x) \leq (1-x)^{-1/2} & = p_T(x) + \sum_{k=T+1}^{\infty}  \prod_{j=0}^{k-1} \left( j+\frac{1}{2}\right)\frac{x^k}{k!} \\
& \leq p_T(x) + \sum_{k=T+1}^{\infty} x^k \\
& \leq p_T(x) + \frac{(1-\eta)^{T+1}}{\eta}.
\end{align*}
Hence, it holds that
\[
(uI-A)^{-1/2}= u^{-1/2} (I- u^{-1}A)^{-1/2} \succeq u^{-1/2} p_T(u^{-1}A),
\]
and
\[
(uI-A)^{-1/2} \preceq u^{-1/2} \left( p_T(u^{-1}A) +\frac{(1-\eta)^{T+1}}{\eta}\cdot I \right),
\]
since $u^{-1}A\preceq (1-\eta) I$. Notice that $u^{-1/2}I \preceq (uI-A)^{-1/2}$, and therefore
\[
(uI-A)^{-1/2} \preceq  u^{-1/2} p_T(u^{-1}A) + \frac{(1-\eta)^{T+1}}{\eta}\cdot (uI-A)^{-1/2}.
\]
Setting $T=\frac{c\log (1/(\eps \eta))}{\eta}$ for some constant $c$ and defining $S_u=u^{-1/2} p_T(u^{-1}A)$ gives us that
\[
S_u\approx_{\eps/10} (uI-A)^{-1/2}.
\]

Now for the second statement. Our construction of $S_{\ell}$ is based on the case distinction ($\ell> 0$, and $\ell\leq 0$).

Case~(1): $\ell>0$. Notice that
\[
(A-\ell I)^{-1/2} =A^{-1/2} (I-\ell A^{-1})^{-1/2},
\]
 and
 \[
p_T(\ell A^{-1}) \preceq \left(I-\ell A^{-1} \right)^{-1/2} \preceq p_T\left(\ell A^{-1}\right) + \frac{(1-\eta/2)^{T+1}}{\eta/2}\cdot I.
 \]
Using the same analysis as before,  we have that
\[
A^{-1/2} (I-\ell A^{-1})^{-1/2} \approx_{\eps/10} A^{-1/2} p_T(\ell A^{-1}).
\]
By defining $S_{\ell} =A^{-1/2} p_T(\ell A^{-1})$, i.e., $A'=A$ and $q\left(  (A')^{-1}\right)=p_T(\ell A^{-1})$, we have that
\[
S_{\ell} \approx_{\eps/10} (A-\ell I)^{-1/2}.
\]

Case~(2): $\ell\leq 0$. We look at the matrix
\[
A-\ell I= L^{-1/2}\tilde{L}L^{-1/2} -\ell I=L^{-1/2} (\tilde{L}-\ell L) L^{-1/2}.
\]
 Notice that $\tilde{L}-\ell L$ is a Laplacian matrix, and
 hence this reduces to the case of  $\ell=0$, for which we  simply set $S_{\ell}=(A-\ell I)^{-1/2}$. Therefore, we can write $S_{\ell}$
as a desired form, where   $A'=A-\ell I$ and polynomial $q=1$.
\end{proof}

\lemref{procedure2} below shows  how to estimate $v_i^{\rot} (uI-A)^{-1} v_i$, and $v_i^{\rot} (A-\ell I)^{-1}v_i$, for all $v_i$ in nearly-linear time.

\begin{lem}\label{lem:procedure2}
Let $A = \sum_{i=1}^{m} v_i v_i^\rot$, and suppose that $A$ satisfies Assumption~\ref{assump1}. Then,
  we can compute
  $\{r_i\}_{i=1}^m$ and $\{t_i\}_{i=1}^m$ in $\tilde{O}\left(\frac{m}{\eps^2\eta}\right)$ time such that
  \[
  (1-\eps) r_i\leq  v_i^{\rot}(uI-A)^{-1}v_i \leq (1+\eps) r_i,\]
   and
   \[
   (1-\eps)t_i\leq v_i^{\rot}(A-\ell I)^{-1}v_i \leq (1+\eps)t_i.\]
\end{lem}

\begin{proof}
Define $u_i = L^{1/2} v_i$ for any $1\leq i\leq m$. By \lemref{procedure1}, we have that
\begin{align*}
v_i^{\rot}(uI-A)^{-1} v_i &\approx_{3\eps/10} \| p(A)v_i \|^2 \\
& = \left\| p\left(L^{-1/2} \tilde{L} L^{-1/2} \right) L^{-1/2} u_i \right\|^2 \\
& =\left\| L^{1/2} p\left(L^{-1} \tilde{L}\right) L^{-1} u_i \right\|^2.
\end{align*}
Let $L=B^{\rot}B$ for some $B\in\mathbb{R}^{m\times n}$. Then, it holds that
\[
v_i^{\rot}(uI-A)^{-1} v_i \approx_{3\eps/10} \left\|B p\left(L^{-1}\tilde{L}\right) L^{-1}u_i \right\|^2.
\]
We invoke the Johnson-Lindenstrauss
Lemma and  find a random matrix $Q\in\mathbb{R}^{O(\log n/\eps^2)\times m}$: With high probability, it holds  that
\[
v_i^{\rot}(uI-A)^{-1} v_i \approx_{4\eps/10 }\left\| QBp\left(L^{-1}\tilde{L}\right) L^{-1} u_i \right\|^2.
\]
We
 apply  a nearly-linear time Laplacian solver to compute $\left\|QBp\left(L^{-1}\tilde{L}\right) L^{-1} u_i \right\|^2$ for all $\{u_i\}_{i=1}^m$ up to $(1\pm\eps/10)$-multiplicative error in time $\tilde{O}\left(\frac{m}{\eps^2\eta}\right)$. This gives the desired $\{r_i \}_{i=1}^m$.

The computation for  $\{t_i\}_{i=1}^m$ is similar. By \lemref{procedure1}, it holds for any $1\leq i\leq m$ that
\begin{align*}
v_i^{\rot}(A-\ell I)^{-1}v_i &\approx_{3\eps/10} \left\| (A')^{-1/2} q((A')^{-1})v_i \right\|^2 \\
& = \left\| (A')^{-1/2} q\left(L^{1/2} (L')^{-1}L^{1/2}  \right)L^{-1/2} u_i \right\|^2 \\
& = \left\| (A')^{-1/2} L^{-1/2} q(L(L')^{-1} )u_i \right\|^2.
\end{align*}
Let $L'=(B')^{\rot}(B')$ for some $B'\in\mathbb{R}^{m\times n}$. Then, it holds that
\begin{align*}
v_i^{\rot} (A-\ell I)^{-1} v_i &\approx_{3\eps/10} \left\|  (L')^{-1/2} q\left( L(L')^{-1} \right)u_i \right\|^2 \\
& =  \left\|  (L')^{1/2} (L')^{-1} q\left( L(L')^{-1} \right)u_i \right\|^2 \\
& = \left\|  B' (L')^{-1} q\left( L(L')^{-1} \right)u_i \right\|^2.
\end{align*}
We invoke the Johnson-Lindenstrauss Lemma and a nearly-linear time Laplacian solver as before
to obtain  required $\{t_i\}_{i=1}^m$. The total runtime is  $\tilde{O}\left(\frac{m}{\eta\eps^2}\right)$.
\end{proof}

\lemref{procedure3} shows that how to approximate $\lambda_{\min}(uI-A)$ and $ \lambda_{\min}(A-\ell I)$ in nearly-linear time.

\begin{lem}\label{lem:procedure3}
Under Assumption~\ref{assump1}, 
we can compute values
 $\alpha,\beta$ in  $\tilde{O}\left(\frac{m}{\eta\eps^3} \right)$ time such that
\[
(1-\eps)\alpha \leq \lambda_{\min}(uI-A) \leq (1+\eps) \alpha
\]
and
\[
(1-\eps)\beta \leq \lambda_{\min}(A-\ell I) \leq (1+\eps) \beta.
\]
\end{lem}

\begin{proof}
By \lemref{procedure1}, we have that $S_u \approx_{\eps/10} (uI-A)^{-1/2}$. Hence, $ \lambda_{\max}(S_u)^{-2} \approx_{3\eps/10} \lambda_{\min}(uI-A)$, and it suffices to estimate $\lambda_{\max}(S_u)$.  Since
\[
\lambda_{\max}(S_u) \leq \left( \tr\left(S_u^{2k}\right) \right)^{1/2k} \leq n^{1/2k} \lambda_{\max}(S_u),
\]
by picking $k=\log n/\eps$ we have that $\left( \tr(S_u^{2k}) \right)^{1/2k} \approx_{\eps/2} \lambda_{\max}(S_u)$.  Notice that
\[
\tr\left( S_u^{2k} \right)  = \tr\left( p^{2k} \left(L^{-1/2}\tilde{L} L^{-1/2} \right)\right)  = \tr\left( p^{2k}\left(L^{-1} \tilde{L}\right)  \right).
\]
Set $\tilde{L} =\tilde{B}^{\rot}\tilde{B}$ for some matrix $\tilde{B}\in\mathbb{R}^{m\times n}$, and we have that $\tr\left( S_u^{2k} \right) = \tr\left( p^{2k}\left( \tilde{B}L^{-1}\tilde{B}^{\rot} \right) \right)$.
Since  we can apply $p^{k}\left( \tilde{B}L^{-1}\tilde{B}^{\rot} \right)$ to vectors in  $\tilde{O}\left(\frac{m}{\eta\eps}\right)$ time,
we  invoke the Johnson-Lindenstrauss Lemma and approximate $\tr\left(S_u^{2k}\right)$
in $\tilde{O}\left(\frac{m}{\eta\eps^3}\right)$ time.

We approximate $\lambda_{\min}(A -\ell I)$ in a similar way. Notice that
\begin{align*}
\tr\left(  S_{\ell}^{4k}\right) & = \tr \left( (A')^{-1/2} q((A')^{-1}) \right)^{4k} \\
& = \tr\left(q( (A')^{-1} ) (A')^{-1} q((A')^{-1} )\right)^{2k}.
\end{align*}
Let $z$ be a polynomial defined by $z(x)=xq^2(x)$ and $L'=(B')^{\rot} (B')$. Then, we have that
\begin{align*}
\tr(S^{4k}_{\ell}) & = \tr\left( z^{2k}((A')^{-1}) \right) = \tr\left( z^{2k}\left(L^{1/2} (L')^{-1}L^{1/2} \right) \right).
\end{align*}
Applying the same analysis as before, we can estimate the trace  in  $\tilde{O}\left(\frac{m}{\eta\eps^3} \right)$ time.
\end{proof}

\end{document}